\documentclass[conference]{IEEEtran}
\IEEEoverridecommandlockouts                              

\usepackage{mathtools} 
\usepackage{algorithm}
\usepackage{algorithmic}

\usepackage[sort,compress]{cite}
\usepackage{amsfonts,dsfont,amssymb,bm}

\usepackage{subcaption}
\usepackage{amsthm}
\usepackage{microtype}

\usepackage{algorithm}
\usepackage{xcolor}

\usepackage{multirow}
\DeclarePairedDelimiter{\abs}{\lvert}{\rvert}

\usepackage{array,graphicx,verbatim,mathrsfs,bbm,mathtools}

\newtheorem{theorem}{Theorem}

\newtheorem{lemma}{Lemma}

\newtheorem{remark}{Remark}

\DeclareMathOperator{\ee}{\mathbb{E}}			
\DeclareMathOperator{\prob}{{\mathds{P}}}			

\usepackage{algorithm}

\usepackage[left=54pt, right=54pt, bottom=54pt, top=54pt]{geometry}

\usepackage{hyperref}

\makeatletter
\newcommand*\doTRANS[2]{\raisebox{\depth}{$\m@th#1\intercal$}}
\makeatother

\makeindex

\usepackage{xparse}
\NewDocumentCommand\AVG{s}
    {\IfBooleanTF#1%
      {\frac{1}{\abs{N  }} \sum_{i \in N  }}
      {\frac{1}{\abs{N^m}} \sum_{i \in N^m}}
    }

\usepackage{xcolor}

\usepackage [english]{babel}
\usepackage [autostyle, english = american]{csquotes}
\MakeOuterQuote{"}

\ifCLASSINFOpdf
\else
\fi

\begin{document}
%
\title{Symmetric Strategies for Multi-Access IoT Network Optimization: A Common Information Approach}
\author{Sagar Sudhakara
\thanks{S. Sudhakara is with the Department of Electrical \& Computer
Engineering, University of Southern California, Los Angeles, CA 90089
(E-mail: sagarsud@usc.edu). }}


%


\maketitle

\begin{abstract}
In the context of IoT deployments, a multitude of devices concurrently require network access to transmit data over a shared communication channel. Employing symmetric strategies can effectively facilitate the collaborative use of the communication medium among these devices. By adopting such strategies, devices collectively optimize their transmission parameters, resulting in minimized collisions and enhanced overall network throughput.

Our primary focus centers on the formulation of symmetric (i.e., identical) strategies for the sensors, aiming to optimize a finite horizon team objective. The imposition of symmetric strategies introduces novel facets and complexities into the team problem. To address this, we embrace the common information approach and adapt it to accommodate the use of symmetric strategies. This adaptation yields a dynamic programming framework grounded in common information, wherein each step entails the minimization of a single function mapping from an agent's private information space to the space of probability distributions over possible actions.

Our proposed policy/method incurs a reduced cumulative cost compared to other methods employing symmetric strategies, a point substantiated by our simulation results.
\end{abstract}


%
\IEEEpeerreviewmaketitle

\section{Introduction}
\subsection{Motivation}

In today's rapidly expanding landscape of Internet of Things (IoT) deployments, the simultaneous operation of a multitude of devices is a commonplace scenario. These devices, ranging from sensors to actuators, play a pivotal role in collecting and disseminating data across interconnected IoT networks. However, as the number of IoT devices continues to grow, so does the contention for network access to transmit data over shared communication channels \cite{khan2012future}. In this context, addressing data collisions becomes a pressing concern.

Data collisions occur when two or more IoT devices attempt to transmit data over the same channel simultaneously, resulting in interference and loss of data. This not only leads to inefficiency in data transmission but also jeopardizes the reliability and responsiveness of the entire IoT network. The consequences of data collisions can be severe, particularly in applications where timely and accurate data exchange is critical, such as industrial automation \cite{misra2022industrial}, healthcare monitoring, smart cities, and precision agriculture.

One of the key challenges in mitigating data collisions within the IoT ecosystem is the need for devices to collaboratively share and manage the communication medium. Traditional communication protocols, while effective in point-to-point or small-scale deployments, may become less efficient as the number of devices grows \cite{gubbi2013internet}. To address this challenge, the adoption of symmetric strategies emerges as a promising solution.

Symmetric strategies involve devices collectively optimizing their transmission parameters in a coordinated manner. This optimization aims to reduce the likelihood of data collisions by ensuring that devices share common rules when accessing the channel. Such strategies, when effectively implemented, can significantly enhance overall network throughput and reduce the occurrence of data collisions.

In our paper, our primary focus revolves around the formulation and implementation of symmetric (i.e., identical) strategies tailored to IoT sensors. These strategies aim to optimize a finite horizon team objective. This concentration on symmetric strategies introduces novel dimensions and complexities into the problem, necessitating innovative approaches for resolution. To tackle these challenges, we turn to the common information approach, which enables us to adapt and accommodate the use of symmetric strategies effectively.

By leveraging the common information approach \cite{9992871}, we construct a dynamic programming framework rooted in common information. In this framework, each step involves the minimization of a function that maps from an agent's private information space to the space of probability distributions over possible actions. This approach provides a robust foundation for devising efficient and collision-avoidance strategies for IoT networks with multiple devices.

The culmination of our research is a proposed policy and methodology that has demonstrated its efficacy in practice. Through extensive simulation results, we have validated that our approach incurs a reduced cumulative cost compared to other methods employing symmetric strategies. This not only underscores the viability of our approach but also its potential to significantly enhance the efficiency and reliability of Multi-IoT networks.

Our paper endeavors to contribute meaningfully to the advancement of Multi-IoT networks by offering innovative solutions for mitigating data collisions and improving network efficiency. The implications of our work extend to a diverse array of applications within the IoT ecosystem, ultimately fostering enhanced data reliability and responsiveness in the era of interconnected devices.
\begin{figure}
      \centering
      \includegraphics[width = 0.40\textwidth]{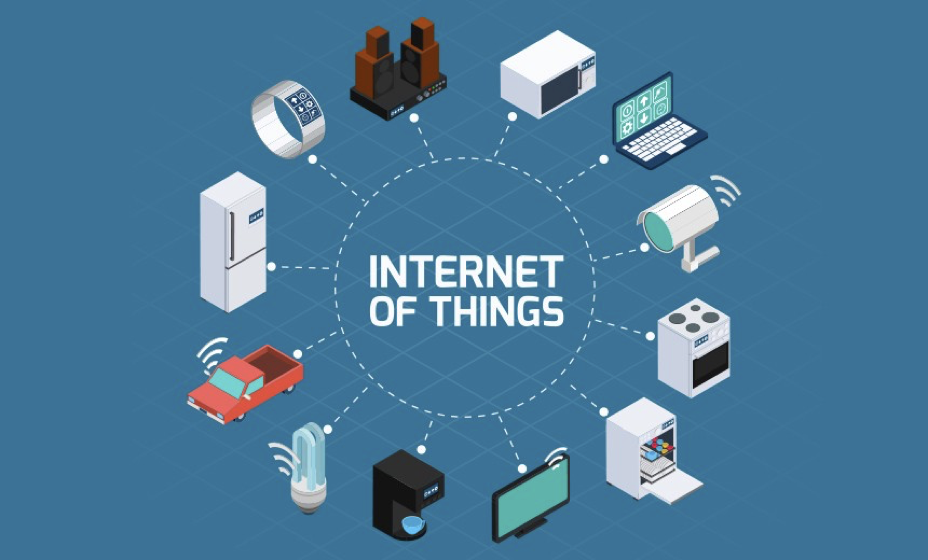}
     \caption{Mitigating Collisions in Multi-IoT Networks.}
\end{figure}
\subsection{Previous Work}
Addressing data collisions in multi-access, single-channel reuse conditions within Multi-IoT networks involves a diverse array of methods and tools. These encompass collision avoidance techniques such as CSMA/CA and listen-before-talk \cite{ieee2007ieee}, where devices actively monitor channel activity prior to transmitting. Time Division Multiple Access (TDMA) \cite{stallings2007data} allots specific time slots for device communication, reducing overlap. Frequency Hopping Spread Spectrum (FHSS) \cite{skylar2015digital} permits devices to switch between channels at designated intervals, minimizing interference. Further enhancements come from randomized backoff times and collision detection with retries \cite{comer2006internetworking}. The central IoT gateway or coordinator plays a pivotal role, managing channel access, scheduling transmissions, and prioritizing critical data via Quality of Service (QoS) mechanisms \cite{li2015internet}. These methods extend to interference mitigation, network synchronization, and machine learning \cite{alsheikh2014machine}, adjusting parameters adaptively for collision reduction. Additionally, channel sensing and dynamic channel allocation optimize IoT communication \cite{akyildiz2002survey}. Moreover, load balancing distributes devices across multiple channels, mitigating congestion. While these strategies are applicable in diverse contexts such as smart homes, industrial IoT, logistics, healthcare, and environmental monitoring, the unique focus on probabilistic symmetric strategies tailored to IoT sensors is notably absent in these approaches. These symmetric strategies, aimed at optimizing finite horizon team objectives, introduce innovative nuances into the problem. To address these challenges, the common information approach is employed, enabling the effective adaptation and utilization of symmetric strategies.
\section{Problem Formulation}\label{sec:problem_formulation}
We delve with a scenario where a group of $n$ IoT sensors operates in a slotted multiple access setting. Time is divided into discrete slots, and within each slot, the sensors can either transmit data or remain idle.
In each time slot, various events occur based on the sensors' activities. These events include:
\begin{itemize}
    \item Successes: If a sensor successfully transmits its data without interference, it is considered a success event. 
    \item Idle Periods: During some time slots, certain sensors may not have any data to transmit. In such cases, they remain idle, indicating no transmission activity.
    \item Collisions: When multiple sensors attempt to transmit data simultaneously within the same time slot, collisions may occur. 
\end{itemize}

Every sensor, denoted by agent $i$, possesses a distinct mode $M^i$ selected from the set $\mathcal{M}$, where $i$ spans the range $\{1, 2, 3, \dots, n\}$. The mode $M^i$ attributed to an agent is a time invariant random variable. The collective set of modes is denoted as $M$, encompassing $(M^1, M^2, \dots, M^n)$.

Agent mode $M^i$ is characterized by three possibilities: Aggressive ($Ag$), Passive ($Pa$), or a selection made by the designer ($De$). In the Aggressive mode, an agent exhibits a high transmission probability denoted as $\alpha$. Conversely, the Passive mode is characterized by a low transmission probability, represented by $\beta$.

The set of possible actions is formally defined as $\mathcal{U} = \{0, 1\}$, wherein the specific control action undertaken by sensor $i$ during time instance $t$ is symbolized by $U^i_t \in \mathcal{U}$. The entirety of control actions for the ensemble of multi IoT sensors is concisely captured as $U_t = (U^1_t, U^2_t, \dots, U^n_t)$.


 \textit{Information structure, randomization, and symmetric strategies:}

At the commencement of time $t$, each IoT sensor $i$, possesses access to its individual mode, its historical sequence of actions, and the action histories of all other sensors. Consequently, the entirety of information accessible to sensor $i$ at time $t$ is concisely represented as:

\begin{equation}
I^i_{t}=\{M^i,U_{1:t-1}\}.
\end{equation}

The information at the disposal of sensor $i$ ($i=1,2,\dots,n$) during time $t$ encompasses two distinct components:
\begin{enumerate}
    \item Common  information $C_t$ - This information is shared among all sensors\footnote{It's important to note that $C_t$ doesn't need to encapsulate the complete scope of information shared by all sensors; it solely excludes information that isn't universally available to all sensors.}. $C_t$ takes values in the set $\mathcal{C}_t$.
    \item Private information $P^i_t$ -Any information exclusively accessible to sensor $i$ at time $t$ that isn't encompassed within $C_t$ is encapsulated within $P^i_t$. The potential values of $P^i_t$ reside within $\mathcal{P}_t$. (Note that the space of private information is the same for all the sensors.) We use $P_t$ to denote  $(P^1_t, P^2_t \dots  P^n_t)$.
\end{enumerate}

$C_t$ should be viewed as an ordered list (or row vector) of some of the system variables that are known to all the sensors. Similarly, $P^i_t$ should be viewed as an ordered list (or row vector).
The common and private information available to sensor $i$ at time $t$  is given by 

\begin{align}
   &C_{t}=U_{1:t-1}; ~P^i_t = M^i.
\end{align}

Sensor $i$ harnesses the accessible information to formulate its transmission decision at time $t$. In the context of multiple sensors, it is a widely accepted practice to confine agents to deterministic strategies without compromising optimality \cite{yuksel2013stochastic}. Nonetheless, given that the sensors in our configuration are constrained to employ symmetric strategies, the introduction of randomization can be beneficial, as demonstrated in an example provided in \cite{sudhakara2023optimal}.
Sensor $i$ has the flexibility to introduce randomness when making its transmission choice. Specifically, at each time $t$, sensor $i$ employs its information $I^i_{t}$ to select a probability distribution $\delta{U}^i_t$ across the transmission decision space $\mathcal{U}$. This is mathematically expressed as:

\begin{equation}\label{commact} \delta{U}^i_t = g^i_t(I^i_{t}), \end{equation}
Here, $g^i_t$ stands as a mapping from $\mathcal{P}_t$ to $\Delta(\mathcal{U})$, where $\Delta(\mathcal{U})$ represents the set of probability distributions over $\mathcal{U}$. Consequently, the control action $U^i_t$ is chosen in a random manner based on the designated distribution, i.e., $U^i_t \sim \delta{U}^i_t$. The primary motivation underlying the use of randomized strategies is rooted in the constraint imposed on our problem to adhere to symmetric strategies, as illustrated in reference [2]. 
The function $g^i_t$ assumes the role of sensor $i$'s transmission strategy during time $t$. The collection of these functions, $g^i:= (g^i_1, \ldots, g^i_T)$, collectively defines sensor $i$'s overarching transmission strategy. The comprehensive set of all possible strategies for sensor $i$ is denoted as $\mathcal{G}$.

We employ the notation $(g^1,g^2, \cdots, g^n)$ to symbolize the strategies adopted by sensor 1, sensor 2, ..., sensor $n$ correspondingly. Our primary focus centers on the realm of \emph{symmetric strategies}, which implies that all sensors adopt an identical transmission strategy. In the context of symmetric strategies, the superscript $i$ in $g^i$ is omitted, and a symmetric strategy pair is succinctly represented as $(g,g,\cdots,g)$.


\subsection{Strategy optimization problem for Multi-IoT Sensors}
During each time instance $t$, the system incurs a cost $k_t(M^1,\dots,M^n,U^1_t\dots U^n_t)$, which hinges on the modes and transmission actions of all sensors. The system operates across a time horizon of duration $T$. The cumulative expected cost associated with a strategy pair $(g^1,g^2,\cdots,g^n)$ is expressed as:

\begin{equation}\label{eq:cost2} J(g^1,g^2,\cdots,g^n):=\mathbb{E}^{(g^1,g^2,\cdots,g^n)}\left[\sum_{t=1}^{T}k_t(M,{U}_t)\right]. \end{equation}
In this context, the system cost at time $t$ is represented by the expression:

\begin{equation}
    k_t(M,{U}_t)=
    \begin{cases}
    0, ~ \text{if only one of} ~ u^i_t=1\\
    1, ~ \text{otherwise} .
  \end{cases} \label{strategy2}
\end{equation}

The cost function indicates a successful transmission when solely one sensor transmits at a given time $t$.

The primary objective is to identify a symmetric strategy pair that minimizes the overall expected cost among all possible symmetric strategy pairs. In other words, the goal is to find a strategy $g \in \mathcal{G}$ such that:
\begin{equation}\label{eq:optimalg}
    J(g,g,\cdots,g) \leq J(h,h, \cdots,h),~ ~\forall h \in \mathcal{G}.
\end{equation}

\begin{remark}
Our premise assumes that the randomization carried out by each sensor remains independent over time and uncorrelated with the other sensor's actions \cite{9992871}. This is achieved by each sensor having access to independent and identically distributed (i.i.d.) random variables $K^i_{1:T}$ uniformly distributed over the interval $(0, 1]$. These variables, $K^1_{1:T}$ and $K^2_{1:T}$, are mutually independent as well as independent from the primitive random variables. Moreover, each sensor leverages a mechanism $\kappa$ that combines $K^i_t$ and the distribution $\delta{U}^i_t$ over $\mathcal{U}$ to generate a random action according to the specified distribution. Thus, sensor $i$'s action at time $t$ is formulated as $U^i_t = \kappa(\delta{U}^i_t, K^i_t)$.
\end{remark}

\begin{remark}
It can be shown that the strategy space $\mathcal{G}$ is a compact space and that $J(g,g,\cdots,g)$ is a continuous function of $g \in \mathcal{G}$. Thus, an optimal $g$ satisfying \eqref{eq:optimalg} exists.
\end{remark}


\section{Common information approach}\label{sec:CI_approach}

In the context of the Multi IoT Sensor Problem, we embrace the common information approach introduced in \cite{nayyar2013decentralized}. 
 This methodology reframes the decision-making challenge from the standpoint of a coordinator endowed with knowledge of the common information. At each time step, the coordinator formulates prescriptions that translate each sensor's private information into its corresponding transmission action. The behavioral action of each sensor within this framework materializes as the prescription evaluated using the current realization of its private information. Notably, since the problem mandates symmetric strategies for all sensors, we demand that the coordinator choose \emph{identical prescriptions for every sensor}. To provide a precise description, let $\mathcal{B}_t$ designate the space encompassing all functions from $\mathcal{P}_t$ to $\Delta(\mathcal{U})$. Within this context, $\Gamma_t \in \mathcal{B}_t$ signifies the prescription elected by the coordinator at time $t$. Consequently, the behavioral action of sensor $i$ takes the form: $\delta U^i_t = \Gamma_t(P^i_t)$.

Aligning with the problem's requirements, sensor $i$'s action $U^i_t$ manifests through independent randomization governed by the distribution $\delta U^i_t$. The coordinator's prescription selection at time $t$ hinges on the common information available at that juncture and the history of prior prescriptions. Hence, the prescription is denoted as:

\begin{equation}
    \Gamma_t = d_t(C_t, \Gamma_{1:t-1}),
\end{equation}

Where $d_t$ constitutes a mapping from $\mathcal{C}_t \times \mathcal{B}_1 \ldots \times \mathcal{B}_{t-1}$ to $\mathcal{B}_t$. The ensemble of mappings $d:= (d_1,\ldots,d_T)$ is acknowledged as the coordination strategy. The coordinator's objective revolves around selecting a coordination strategy that effectively minimizes the total expected cost across the finite time horizon:

\begin{equation}
    \mathcal{J}(d) := \ee^{d}\left[\sum_{t=1}^{T}k_t(M,{U}_t)\right].
\end{equation}

The subsequent lemma serves to establish the equivalence between the coordinator's reformulation and the original problem. The integration of identical prescriptions by the coordinator acts as the linchpin for connecting the coordinator's strategy to symmetric strategies for the agents in the original problem.

\begin{lemma}
The original Multi IoT Sensor Problem and the reformulated coordinator's problem are equivalent in the ensuing sense:

(i) For any symmetric strategy pair $(g,g)$, consider the subsequent coordination strategy:
\[ d_t(C_t) = g_t(\cdot, C_t).\]
In this scenario, $J(g,g) = \mathcal{J}(d)$.

(ii) Conversely, for any coordination strategy $d$, define the symmetric strategy pair as follows:
\[ g_t(\cdot, C_t) = d_t(C_t, \Gamma_{1:t-1} ),\]
where $\Gamma_k = d_k(C_k, \Gamma_{1:k-1})$ for $k=1,\ldots, t-1$.
\end{lemma}
\begin{proof}
    The proof hinges on Proposition 3 of \cite{nayyar2013decentralized}, coupled with the realization that the use of identical prescriptions by the coordinator corresponds to the adoption of symmetric strategies within the original problem.

In this manner, the common information approach bridges the gap between the original Multi IoT Sensor Problem and the perspective of a coordinator employing prescriptions based on the shared information to optimize the sensors' strategies.
\end{proof}



We now proceed with finding a solution for the coordinator's problem. As shown in \cite{nayyar2013decentralized}, the coordinator's belief on $P_t$ can serve as its information state  (sufficient statistic) for selecting prescriptions. At time $t$, the coordinator's belief is given as:
\begin{align}\label{coord:prob1_belief}
    &\Pi_t(p)=\prob(P_t =p|C_t, \Gamma_{1:(t-1)}),
\end{align}
for all $ p \in \mathcal{P}_t \times \mathcal{P}_t \dots \times \mathcal{P}_t $. The belief can be sequentially updated by the coordinator as described in  Lemma \ref{lemma:belief_update} below. The lemma follows from arguments similar to those in Theorem 1 of \cite{nayyar2013decentralized}.

\begin{lemma} \label{lemma:belief_update}
For any coordination strategy $d$, the coordinator's belief $\Pi_t$ evolves almost surely as
\begin{equation}
    \Pi_{t+1} = \eta_t(\Pi_t, \Gamma_t, U_{t}),
\end{equation}
where $\eta_t$ is a fixed transformation that does not depend on the coordination strategy. 
\end{lemma}
\begin{proof}
We know that the coordinator's belief at time $t$, is given by:
\begin{align}\label{coord:prob1_belief}
    &\Pi_t(m)=\prob(P_t =m|U_{1:t-1}, \Gamma_{1:(t-1)}),
\end{align}
$\forall ~$m$ \in \{Ag, Pa, De\}$. At time $t+1$,
\begin{align}\label{coord:prob1_belief}
    &\Pi_{t+1}(m):=\prob(P_{t+1}={m}|U_{1:t}, \Gamma_{1:t})\notag\\
    &=\frac{\prob(m,u_t|u_{1:t-1}, \gamma_{1:t})}{\sum_{\tilde{m}}\prob(\tilde{m},{u}_t|u_{1:t-1}, \gamma_{1:t})}\notag\\
    &=\frac{\Pi_{i=1}^n \gamma_t(m^i,u_{1:t-1};u^i_t)\Pi_t(m)}{\sum_{\tilde{m}} \Pi_{i=1}^n \gamma_t(\tilde{m}^i,{u}_{1:t-1};{u}^i_t)\Pi_t(\tilde{m})}
\end{align}
\end{proof}
Using  the results in \cite{nayyar2013decentralized}, we can write a dynamic program for the coordinator's  problem. Recall that  $\mathcal{B}_t$ is the space of all functions from $\mathcal{P}_t$ to $\Delta(\mathcal{U})$. For a $\gamma \in \mathcal{B}_t$ and $p \in \mathcal{P}_t$, $\gamma(p)$ is a probability distribution on $\mathcal{U}$. Let $\gamma(p;u)$ denote the probability assigned to $u \in \mathcal{U}$ under the probability distribution $\gamma(p)$.

\begin{theorem}\label{thm:dp}
The value functions  for the coordinator's dynamic program are as follows: Define $V_{T+1}(\pi_{T+1}) =0$ for every $\pi_{T+1}$. For $t \leq T$ and  for any realization $\pi_t$ of $\Pi_t$, define
\begin{multline}
    V_t(\pi_t) = \min_{\gamma_t \in \mathcal{B}_t} \mathbb{E}[k_t(M, U_t) + \\
    V_{t+1}(\eta_t(\pi_t, \gamma_t, U_{t})) | \Pi_t = \pi_t, \Gamma_t = \gamma_t]\label{eq:DP1_new}
\end{multline}
Expanding on equation \eqref{eq:DP1_new},
\begin{equation}
\begin{split}
    \min_{\gamma_t \in \mathcal{B}_t} \Biggl[ &\sum_{m} \sum_{u_{1:t}} k_t(m, u_t)\pi_t(m)\Pi_{i=1}^n\gamma_t(m^i,u_{1:t-1};u^i_t) \\
    &+V_{t+1}(\eta_t(\pi_t, \gamma_t, u_{t}))\pi_t(m)\Pi_{i=1}^n\gamma_t(m^i,u_{1:t-1};u^i_t) \Biggr]
\end{split}
\end{equation}




 The coordinator's optimal strategy is to pick the minimizing prescription  for each time and each $\pi_t$.
\end{theorem}

Let $\Xi_t(\pi_t)$  be a minimizer of the value function in \eqref{eq:DP1_new}. Using $\Xi_t(\pi_t)$ obtained from the dynamic program in Theorem \ref{thm:dp}, we can construct a symmetric strategy pair $({g}^{*},\cdots,{g}^{*})$ as described in Algorithm \ref{alg:example}. This strategy pair minimizes the overall expected cost among all possible symmetric strategy pairs.
\begin{algorithm}[h]
  \caption{Strategy ${g}^{*}$ for IoT Sensor $i$}
  \label{alg:example}
\begin{algorithmic}
    \STATE Input: $\Xi_t(\cdot)$ obtained from DP for all $t$ 
  \FOR{$t=1$ {\bfseries to} $T$}
  
  \STATE Current information: $C_t,M^{i}$
  \STATE Update CIB $\Pi_{t+1} = \eta_t(\Pi_{t}, \Xi_{t}(\Pi_t),U_t)$ 
  \COMMENT{If $t=1$, Initialize CIB $\Pi_t$ using $C_1$}
  \STATE Get prescription $\Gamma_t = \Xi_t(\Pi_t)$ 
  \STATE Select communication action $U_t^{i} = \Gamma_t(M^{i}) $
  \ENDFOR
\end{algorithmic}
\end{algorithm}

We emphasize that the common information based dynamic program is solved offline. Therefore, the solution $\Xi_t$ to this dynamic program is known to both agents before the agents start operating in their environment. During their operation, this solution is used in a decentralized manner by the agents to select their actions as described in Algorithm \ref{alg:example}.

\section{Experiments}
For the sake of simplicity and to facilitate a clearer grasp of the design of symmetric strategies, we have constructed our experiment using only two sensors. This deliberate reduction in complexity allows us to provide a clearer insight into strategy designing when sensors are in different modes, such as Aggressive, Passive, or Designer mode.

\textbf{Problem Setup} Consider a scenario involving two sensors. Each sensor is capable of operating in one of three modes: Aggressive (Ag), Passive (Pa), or Designer (De). In Aggressive mode, a sensor has a high transmission probability ($\alpha$) near 1. In Passive mode, the transmission probability ($\beta$) is low, almost approaching 0. Our objective is to devise symmetric strategies for the sensors, accounting for the fact that one sensor's mode is private to the other and vice versa. The team incurs a cost of 0 for successful transmissions and a cost of 1 for collisions or idle channel situations.

\textbf{Implementation} 
To recap, let's revisit the concept of $\gamma(p)$, which denotes a continuous probability distribution on $\mathcal{U}$. In order to work with this continuous distribution in a discrete manner, we employ quantization. Specifically, we discretize the probability distribution into 21 distinct probability values, each separated by 0.05. This quantization process is crucial to enable our analysis within a discrete framework. The time horizon, labeled as $T$, is set to 10.

The procedure commences with an initial belief matrix $\textbf{\textit{B}}$ at time $t=1$, wherein the rows indicate the mode of agent $1$ and columns indicate the mode of agent $2$. By utilizing equation (21), we progressively update the belief matrix for each subsequent time step $t$. Following the update of the belief matrices, our emphasis shifts towards identifying the unique matrices among all these beliefs. This strategic step prevents redundant recalculations of values, ultimately boosting computational efficiency and streamlining our analysis process.

We solve the Coordinator's Dynamic program given in equation (13) and (14), store the minimum value for every belief ($\pi_t$) that can be achievable at time $t$ from the updated belief matrices, and back-propagate in time to evaluate the Value of all the beliefs reachable at time $t-1$.

The value obtained at time $t=0$ is the cumulative expected cost associated with the symmetric strategy. By executing these steps meticulously, we successfully transform theoretical concepts into practical computations, allowing us to uncover valuable insights and deliver conclusive results.\\ 

\textbf{Results 1:} 
We initialize the belief matrix indicated as $\textbf{\textit{B}}_1$ in Figure \ref{fig:scatter_plot1} with specific conditions: agent $1$ is assumed to be in the designer mode, while agent $2$ can be in either Aggressive or Passive mode, each with a probability of 0.5. (Here, $De^i$ represents agent $i$ being in Designer mode.)

\[
\textbf{\textit{B}}_1 = 
\begin{array}{c|ccc}
    & \text{$De^2$} & \text{$Ag^2$} & \text{$Pa^2$} \\
\hline
\text{$De^1$} & 0 & 0.5 & 0.5 \\
\text{$Ag^1$} & 0 & 0 & 0 \\
\text{$Pa^1$} & 0 & 0 & 0
\end{array}
\]

We observe scatter plots of value versus belief that depict the introduction of new observations (actions). As new observations are integrated, updated belief matrices emerge, each associated with a probability of taking action $U^1_t=1$ for every belief, as demonstrated in the Figure \ref{fig:scatter_plot1} below. The belief matrices indicated in Figure \ref{fig:scatter_plot1} are:

\vspace{3mm}
\begin{tabular}{c|c|c}
  \text{$\textbf{\textit{B}}_1$} & \text{$\textbf{\textit{B}}_{2}$} & \text{$\textbf{\textit{B}}_{3}$} \\
\hline
  $\begin{bmatrix}
    0 & 0.5 & 0.5 \\
    0 & 0 & 0\\
    0 & 0 & 0
  \end{bmatrix}$ &
  $\begin{bmatrix}
    0 & 1 & 0 \\
    0 & 0 & 0\\
    0 & 0 & 0
  \end{bmatrix}$ &
  $\begin{bmatrix}
    0 & 0 & 1 \\
    0 & 0 & 0\\
    0 & 0 & 0
  \end{bmatrix}$ \\
\end{tabular}

\begin{figure}[htb!]
    \centering
    \includegraphics[width=0.50\textwidth]{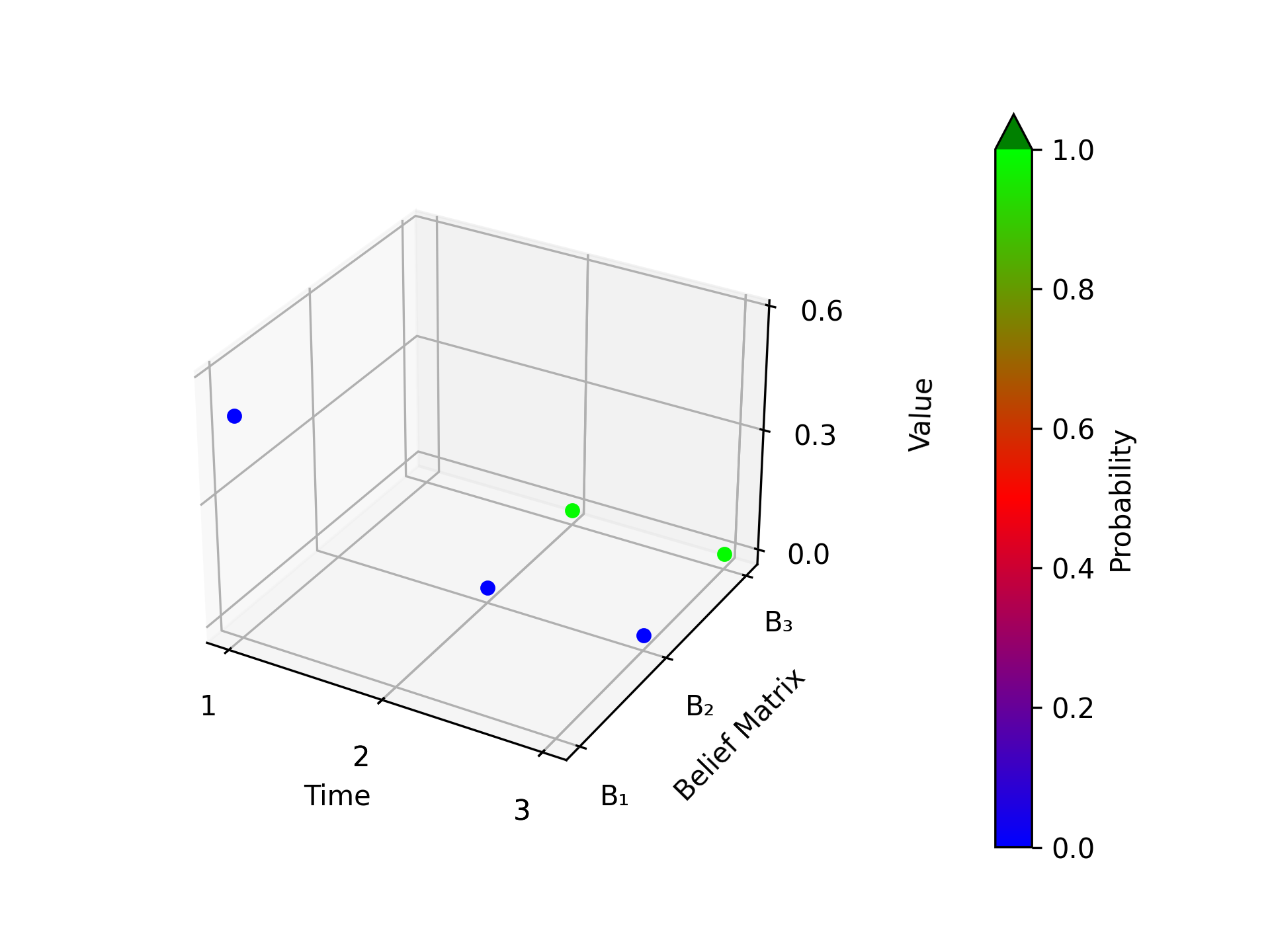} 
    \caption{Scatter Plot of Value vs. Belief}
    \label{fig:scatter_plot1}
\end{figure}
\vspace{3mm}
Beyond time $t=2$, the updated matrices tend to stabilize, enabling sensor $1$ to develop a more refined understanding of sensor 2's mode. This knowledge equips agent $1$ to strategically choose actions that minimize the system cost.\\

\textbf{Results 2:} We present a comprehensive comparative study between our symmetric strategic belief approach and alternative belief-based transmission strategies. This comparison is visually illustrated through a bar plot, effectively capturing the cumulative expected costs.

The alternative belief-based strategies are characterized by parameters represented as $(x,y,z)$. Under the premise that sensor 1 operates in designer mode, it acquires insight into the mode distribution of the other sensor. If the likelihood of the other sensor being in aggressive mode surpasses the threshold $x$ where $(\alpha\leq x \leq 1)$, the designer mode triggers action $(U^1_t)=0$. Similarly, when the probability of the other sensor being in passive mode exceeds the threshold $y$ where $(1-\beta \leq y \leq 1)$, the designer mode opts for action $(U^1_t)=1$. In cases where neither of these conditions holds true, the designer mode of sensor 1 selects action $(U^1_t)=1$ with a probability of $z$ $(\beta \leq z \leq \alpha)$ recall where $\alpha$ is high probability when the sensor is in aggressive mode and $\beta$ is the lowest probability of transmission when the sensor is in passive mode.

\begin{figure*}[h]
    \centering
    \includegraphics[width=0.80\textwidth]{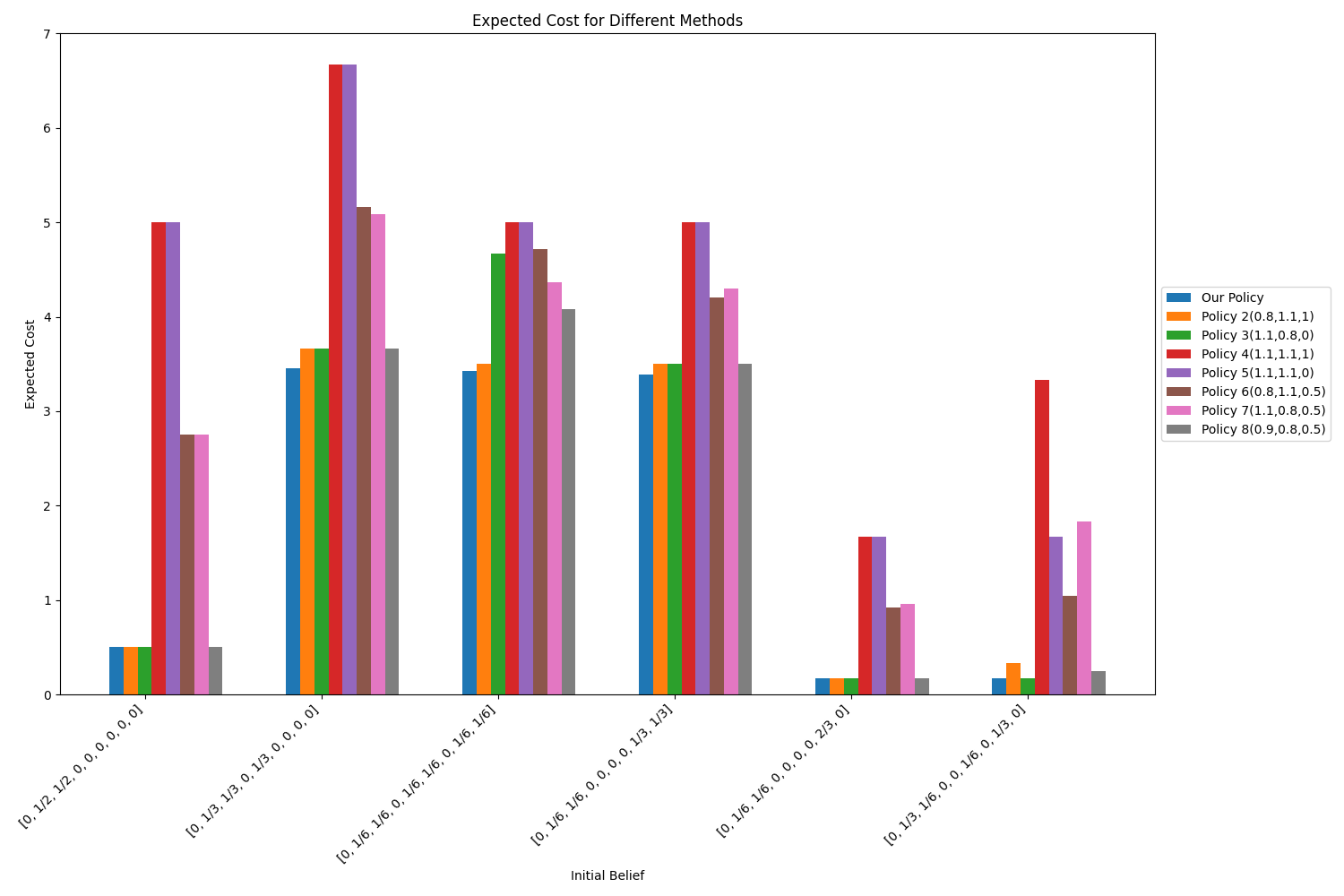} 
    \caption{Scatter Plot of Value vs. Belief}
    \label{fig:scatter_plot}
\end{figure*}
As shown in Table 1, our analysis encompasses diverse initial belief matrices, visually represented as vectors along the x-axis of the bar plot. By scrutinizing the expected costs associated with various belief-based methods characterized by the parameters $(x,y,z)$—as indicated in the table below policy methods—we uncover valuable insights into the effectiveness of our proposed symmetric strategic belief approach.

Remarkably, the table unequivocally displays the superiority of our method in achieving the lowest expected cost across a spectrum of initial beliefs compared to other belief-based methods. The least effective among the belief-based methods consistently adopt either an aggressive or passive stance, as exemplified by methods 4 and 5. The simulation code for these analyses is readily accessible through the following GitHub link \href{https://github.com/Sagarsud93/IoT-Symmetric-Strategy/blob/main/Copy_of_Policy_class_abc.ipynb}.

\begin{table*}[ht]
\centering
\begin{tabular}{|c|c|c|c|c|c|c|c|c|c|}
\hline
\multirow{2}{*}{Initial Belief} & \multirow{2}{*}{Proposed Policy} & \multirow{2}{*}{Policy 2} & \multirow{2}{*}{Policy 3} & \multirow{2}{*}{Policy 4} & \multirow{2}{*}{Policy 5} & \multirow{2}{*}{Policy 6} & \multirow{2}{*}{Policy 7} & \multirow{2}{*}{Policy 8} \\[1ex]
& & (0.8,1.1,1) & (1.1,0.8,0) & (1.1,1.1,1) & (1.1,1.1,0) & (0.8,1.1,0.5) & (1.1,0.8,0.5) & (0.9,0.8,0.5) \\
\hline
$[0, \frac{1}{2}, \frac{1}{2}, 0, 0, 0, 0, 0, 0]$ & 0.5 & 0.5 & 0.5 & 5 & 5 & 2.75 & 2.75 & 0.5 \\[1ex]
\hline
$[0, \frac{1}{3}, \frac{1}{3}, 0, \frac{1}{3}, 0, 0, 0, 0]$ & 3.457 & 3.667 & 3.667 & 6.667 & 6.667 & 5.167 & 5.082 & 3.667 \\[1ex]
\hline
$[0, \frac{1}{6}, \frac{1}{6}, 0, \frac{1}{6}, \frac{1}{6}, 0, \frac{1}{6}, \frac{1}{6}]$ & 3.429 & 3.5 & 4.667 & 5 & 5 & 4.715 & 4.369 & 4.083 \\[1ex]
\hline
$[0, \frac{1}{6}, \frac{1}{6}, 0, 0, 0, 0, \frac{1}{3}, \frac{1}{3}]$ & 3.39 & 3.5 & 3.5 & 5 & 5 & 4.205 & 4.295 & 3.5 \\[1ex]
\hline
$[0, \frac{1}{6}, \frac{1}{6}, 0, 0, 0, 0, \frac{2}{3}, 0]$ & 0.167 & 0.167 & 0.167 & 1.667 & 1.667 & 0.917 & 0.963 & 0.167 \\[1ex]
\hline
$[0, \frac{1}{3}, \frac{1}{6}, 0, 0, \frac{1}{6}, 0, \frac{1}{3}, 0]$ & 0.167 & 0.333 & 0.167 & 3.333 & 1.667 & 1.042 & 1.835 & 0.25 \\[1ex]
\hline
\end{tabular}
\caption{Comparison of cumulative Expected cost for different Policies}
\end{table*}

\section{Conclusion}
In conclusion, the optimization of communication within IoT networks, particularly those with numerous interconnected sensors, is paramount for ensuring efficient network operation. Collaborative decision-making in uncertain environments becomes indispensable in these networks, where simultaneous data transmission poses the risk of collisions and reduced efficiency.

Symmetric strategies, inspired by cooperative multi-agent scenarios, emerge as a promising solution. By deploying identical decision-making protocols among IoT devices, these strategies empower devices to optimize their transmission patterns. This optimization reduces data collisions, ultimately leading to improved network throughput.

The transformative potential of symmetric strategies in IoT networks cannot be understated. These strategies streamline data transmission, enhance network efficiency, and unlock the full potential of IoT networks. Our paper delves into the application of symmetric strategies to multi-access IoT sensor networks, providing valuable insights and empirical evidence on their profound impact.

Our overarching goal is to pave the way for improved coordination among IoT devices, ensuring efficient communication and a seamless future for IoT networks. As the IoT landscape continues to evolve, symmetric strategies stand as a beacon of promise for achieving these crucial objectives.


\section*{Acknowledgment}
The author would like to thank the Annenberg Fellowship and the grants from the Graduate School research award.



%

\bibliographystyle{ieeetr}
\bibliography{IEEEabrv, ref_learning}



\end{document}